\documentclass[a4paper]{elsarticle}
\pdfoutput=1 
\journal{Operations Research Letters}

\usepackage{amsmath,amssymb,natbib,setspace,hyperref,graphicx,soul,setspace}
\usepackage[left=1in,right=1in,top=1.4in,bottom=1.2in]{geometry}
\hypersetup{
	colorlinks=true,
	linkcolor=red,
	citecolor=blue,
	urlcolor=blue,
	bookmarksnumbered=true,
	pdfstartview=
}

\newcommand{\qtext}[2][\quad]{#1\text{#2}#1}
\newcommand{\erfcx}[1]{\mathrm{Erfcx}\left(#1\right)}
\newcommand{\bsvega}{\mathcal{V}_\mathrm{BS}}
\newcommand{\bsdelta}{\mathcal{D}_\mathrm{BS}}
\newcommand{\price}{C_\mathrm{BS}}
\newcommand{\pvega}{C_\mathcal{V}}
\newcommand{\pdelta}{C_\mathcal{D}}
\newcommand{\kk}{k}

\newtheorem{prop}{Proposition}
\newtheorem{lemma}{Lemma}
\newproof{pf}{Proof}
\newtheorem{cor}{Corollary}[prop]
\newtheorem{rem}{Remark}[prop]

\date{October 2, 2024}

\begin{document}

\begin{frontmatter}
\title{Tighter `uniform bounds for Black--Scholes implied volatility'\\ and the applications to root-finding}

\author{Jaehyuk Choi\corref{corrauthor}}
\ead{jaehyuk@phbs.pku.edu.cn}
\address{Peking University HSBC Business School, Shenzhen, China}

\author{Jeonggyu Huh}
\ead{jghuh@skku.edu}
\address{Department of Mathematics, Sungkyunkwan University, Suwon, Korea} 

\author{Nan Su}
\ead{sunan@stu.pku.edu.cn}
\address{Peking University HSBC Business School, Shenzhen, China}

\cortext[corrauthor]{Corresponding author \textit{Tel:} +86-755-2603-0568, \textit{Address:} Peking University HSBC Business School, University Town, Nanshan, Shenzhen 518055, China}

\begin{abstract}
	Using the option delta systematically, we derive tighter lower and upper bounds of the Black--Scholes implied volatility than those in \citeauthor{tehranchi2016bound} [SIAM J. Financ. Math. 7 (2016), 893--916]. As an application, we propose a Newton--Raphson algorithm on the log price that converges rapidly for all price ranges when using a new lower bound as an initial guess. Our new algorithm is a better alternative to the widely used naive Newton--Raphson algorithm, whose convergence is slow for extreme option prices.
\end{abstract}
\end{frontmatter}

\section{Introduction} \noindent
\citet{tehranchi2016bound} found several useful upper and lower bounds of the Black--Scholes implied volatility (IV) given the option premium and strike price. Using the bounds, \citep{tehranchi2016bound} rederived several IV asymptotics known in the literature~\citep{demarco2017shapes,gulisashvili2010asymptotic,tehranchi2009asymptotics}.

The purpose of this paper is two-fold. First, in Section~\ref{s:bounds}, we improve the bounds of \citep{tehranchi2016bound}. The bounds in this paper are summarized and compared with \citep{tehranchi2016bound} as follows:
\begin{gather*}
	\underbrace{-\kk \Big/ \Phi^{-1}\left(\dfrac{c}{1+e^\kk}\right)}_\text{\citep[Prop.~4.6]{tehranchi2016bound}}
	\le\; \underbrace{\vphantom{\bigg|} L_2(c) = d_1^{-1}\circ \Phi^{-1}(c)}_\text{\citep[Prop.~4.9]{tehranchi2016bound}} \,\le\,
	\underbrace{L_3(c) = d_1^{-1} \circ \Phi^{-1}\left(\frac{c\,(c+e^\kk)}{2c + e^\kk - 1}\right)}_\text{This paper (Cor.~\ref{c:l3})} \\
	\le\, \underbrace{\vphantom{\bigg|} L_{U23}(c) = J\circ U_{23}(c)}_\text{This paper (Cor.~\ref{c:l23})} \,\le\, \sigma \,\le\,
	\underbrace{U_{23}(c) = H\left(\min\left(\frac{1 + c}{2},\, c + e^\kk \Phi(-\sqrt{2\kk})\right)\right)}_\text{This paper (Cor.~\ref{c:u23})} \,\le \\
	\underbrace{U_3(c) = -\Phi^{-1}\left(\frac{1 - c}{2}\right) - \Phi^{-1}\left(\frac{1 - c}{2 e^\kk}\right)}_\text{This paper (Prop.~\ref{p:up_new})} \,\le\,
	\underbrace{U_1(c) = -2\,\Phi^{-1}\left(\frac{1 - c}{1 + e^\kk}\right)}_\text{\citep[Prop.~4.3]{tehranchi2016bound}}.
\end{gather*}
We systematically use the bounds of the option delta to discover tighter bounds. We also provide two volatility transforms that change a lower bound into an upper bound~($J(\cdot)$ in Proposition~\ref{p:up2lo}) and a lower bound into a tighter lower bound~($G(\cdot)$ in Proposition~\ref{p:nr_log}). The transforms complement the one in \citep[p.~914]{tehranchi2016bound} that changes an upper bound into a tighter upper bound.

Second, in Section~\ref{s:root}, we propose a new numerical root-finding method for IV. The naive Newton--Raphson (NR) method converges very slowly for out-of-the-money options with very low prices. We propose an alternative NR method based on the log price to solve the problem. We prove that the new NR method is guaranteed to converge to the root from below when the iteration begins from a lower bound. When a new lower bound, $L_3(c)$, is used as an initial guess, numerical experiments demonstrate that the new NR method converges uniformly fast to the root within only a few iterations. 

\section{Tighter bounds of implied volatility} \label{s:bounds}
\subsection{Standardized Black--Scholes formula} \noindent
The undiscounted Black--Scholes option price is given by
$$
C = \theta\left[F\, \Phi(\theta D_1) - K\,\Phi(\theta D_2)\right]
\qtext{for} D_{1,2} = \frac{\log(F/K)}{\Sigma\sqrt{T}} \pm \frac{\Sigma\sqrt{T}}{2},
$$
where $\Sigma$ is the volatility, $F$ is the forward stock price, $K$ is the strike price, $T$ is the time to maturity, $\Phi(\cdot)$ is the standard normal probability function, and $\theta=\pm 1$ indicates the call and put options, respectively. The range of the option value $C$ is
$$ [\,\theta(F-K)\,]^+\le C \le \min(K, F) + [\,\theta(F-K)\,]^+ = \left\{ \begin{matrix}
	K & (\theta=-1) \\
	F & (\theta=+1)
\end{matrix} \right. \;;$$
thus, we standardize $C$ by using
$$ \price(\sigma) = \frac{C - [\,\theta(F-K)\,]^+}{\min(K, F)} =
\Phi\circ d_1(\sigma) - e^\kk\,\Phi\circ d_2(\sigma) \qtext{with}
d_{1,2}(\sigma) = -\frac{\kk}{\sigma} \pm \frac{\sigma}2,
$$
where $\sigma$ is the standard deviation, and $k$ is the log standardized strike:
$$ \sigma = \Sigma \sqrt{T} ,\quad \kk = |\log(F/K)| \ge 0.
$$
Therefore, without loss of generality, it is sufficient to consider an out-of-the-money ($K \ge F$ or $\kk \ge 0$) call option ($\theta=1$) with $F = 1$ and $T = 1$. The standardized option price $\price(\cdot)$ is a bijection between $\sigma\in(0,\infty)$ and $c\in(0, 1)$. The bijection is a monotonically increasing S-shaped function with a unique inflection point at $\sigma=\sqrt{2k}$ (see Figure~\ref{fig:price} (upper left)). Therefore, the IV, $\sigma>0$, always exists for a given price, $0<c<1$. In the remainder of this paper, we let $c$ and $\sigma$ denote the standardized option price and corresponding IV, respectively, satisfying $\price(\sigma) = c$. Note that $d_1(\cdot)$ is a monotonically increasing function; thus, the inverse is well-defined by
\begin{equation} \label{eq:d1inv}
	d_1^{-1}(x) = x + \sqrt{x^2 + 2\kk} \quad (\;=2(x)^+ \qtext{if} \kk=0\;).
\end{equation}

The delta and vega of the standardized price,
$$ \bsdelta(\sigma) = \Phi\circ d_1(\sigma) \qtext{and} \bsvega(\sigma) = \phi\circ d_1(\sigma),
$$
are used later. In addition, $\phi(\cdot)$ is the standard normal density function. Throughout this paper, the functions (e.g., $\price$, $\bsdelta$, $d_1$, and $d_1^{-1}$) are considered as the univariate functions of $\sigma$. The log strike price $\kk$ is a parameter of the functions; thus, it is omitted unless the dependency on $\kk$ is explicitly required.

\begin{figure}[ht!]
	\caption{\label{fig:price} Standardized BS price $\price(\sigma)$, price-to-delta ratio $\pdelta(\sigma)$, log price $\log \price(\sigma)$, and price-to-vega ratio $\pvega(\sigma)$ as functions of $\sigma$. The blue circle marks the inflection point, $\sigma=\sqrt{2k}$, of $\price(\sigma)$.
	}
	\begin{center}
		\includegraphics[width=0.495\linewidth]{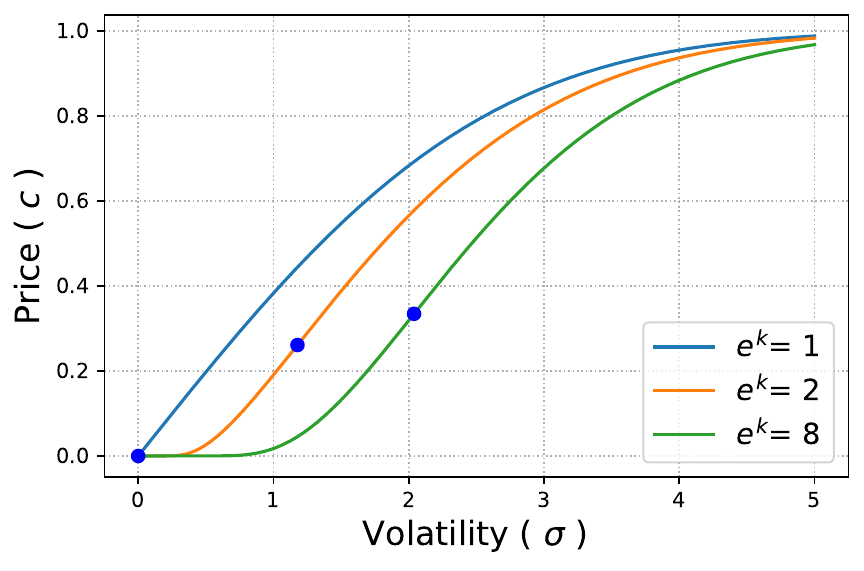}
		\includegraphics[width=0.495\linewidth]{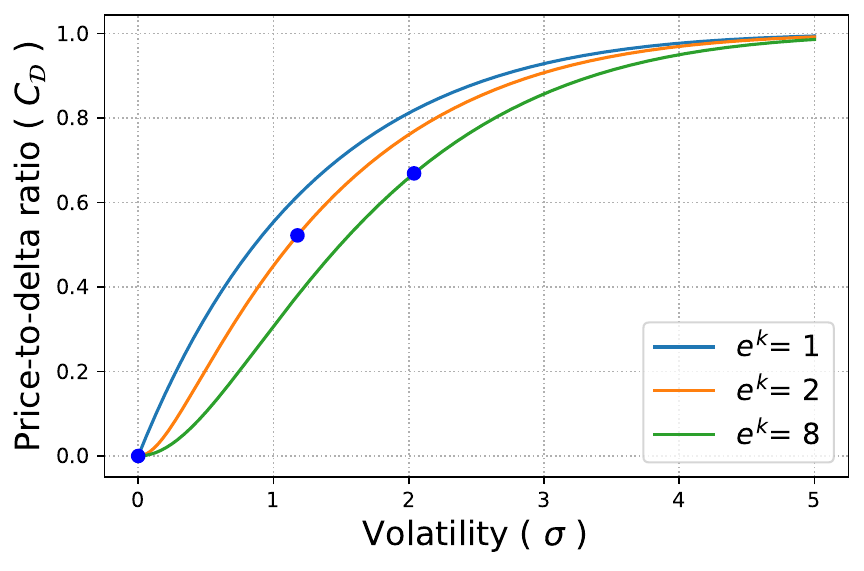}\\
		\includegraphics[width=0.495\linewidth]{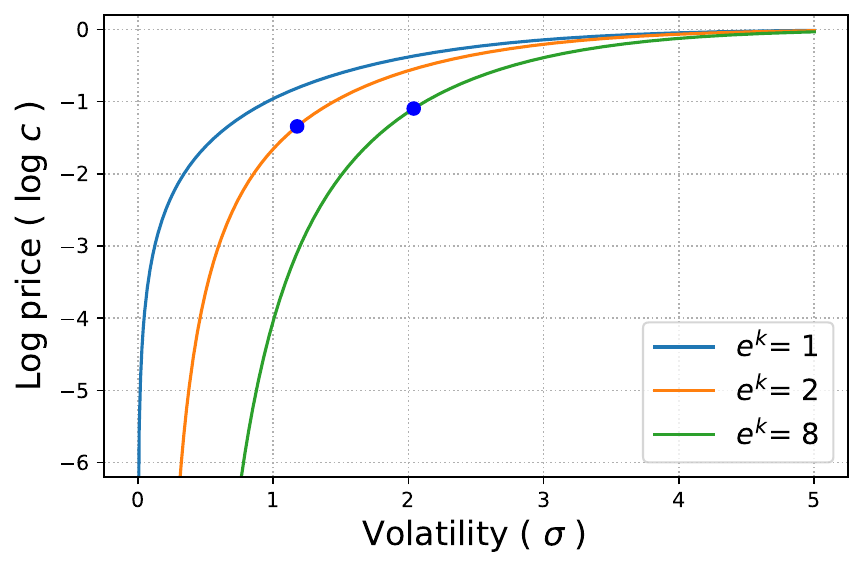}
		\includegraphics[width=0.495\linewidth]{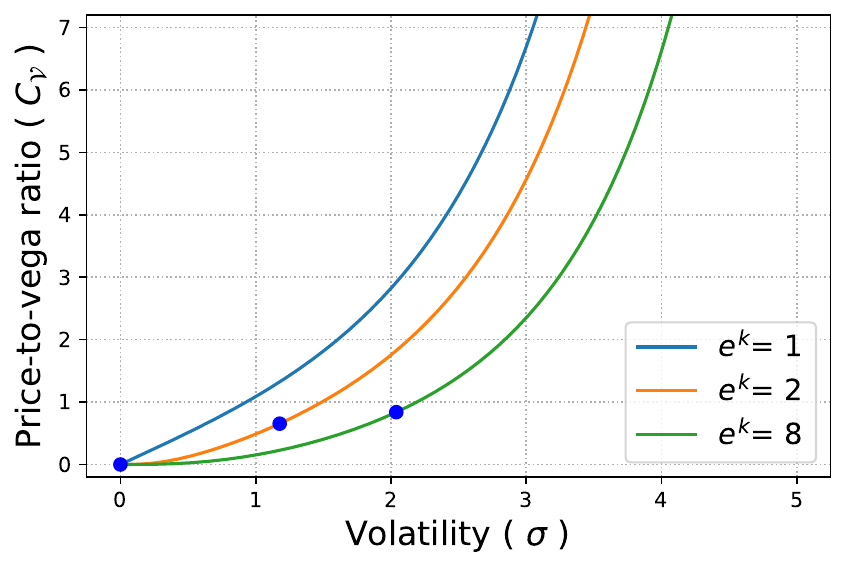}
	\end{center}
\end{figure}

\subsection{Bounds from \citet{tehranchi2016bound}} \label{ss:old} \noindent
We list some results of \citep{tehranchi2016bound} relevant to the new findings.
We also highlight the redundancy among the lower bounds in \citep{tehranchi2016bound}.
\begin{prop}[\mbox{\citet[Prop.~4.3]{tehranchi2016bound}}] \label{p:teh}
For a given option price $c$, the IV $\sigma$ is bounded between $L_1(c)$ and $U_1(c)$ (i.e., $L_1(c)\le \sigma\le U_1(c)$) defined below:
\begin{equation*}
	L_1(c) = -2\,\Phi^{-1}\left(\frac{1 - c}{2}\right)
	\qtext{and} U_1(c) = -2\,\Phi^{-1}\left(\frac{1 - c}{1 + e^\kk}\right).	
\end{equation*}
\end{prop}
\begin{pf}
The lower bound is derived from the fact that an at-the-money ($\kk=0$) option is more valuable than an out-of-the-money option:
\begin{gather*}
	c = \price(\sigma) \le \price(\sigma; \kk=0) = \Phi\left(\sigma/2\right) - \Phi\left(-\sigma/2\right) = 1-2\,\Phi\left(-\sigma/2\right) \\
	L_1(c) = - 2\,\Phi^{-1}\left(\frac{1 - c}{2}\right) = 2\,\Phi^{-1}\left(\frac{1 + c}{2}\right). \label{eq:lower1}
\end{gather*}
For the upper bound, imagine that the option holder can make an exercise decision regardless of the option's actual moneyness at maturity, and that she exercises the option \textit{incorrectly} when $S_T\ge K^*$ which is not necessarily equal to the contractual strike $K$. Then, the option value to the holder is always lower than the \textit{correctly} exercised option price because the payoff under \textit{incorrect} exercise is always lower:
$$ (S_T - K) \cdot 1_{S_T\ge K^*} \le (S_T - K)^+,
$$
where the equality holds only if $K=K^*$. With $K^*=F$ (i.e., incorrectly exercised when $S_T\ge F$), we obtain the upper bound of $\sigma$:
\begin{gather*}
\mathbb{E}\left((S_T - K) \cdot 1_{S_T\ge F}\right) = \Phi(\sigma/2) - e^\kk\,\Phi(-\sigma/2) = 1 - (1 + e^\kk)\,\Phi(-\sigma/2) \le c, \\
U_1(c) = -2\,\Phi^{-1}\left(\frac{1 - c}{1 + e^\kk}\right) = 2\,\Phi^{-1}\left(\frac{c + e^\kk }{1 + e^\kk}\right). \quad \square
\end{gather*}
\end{pf}

\begin{prop}[\mbox{\citet[Prop.~4.9]{tehranchi2016bound}}] \label{p:l2}
	Below, $L_2(c)$ is a lower bound of $\sigma$:
	\begin{equation*}
		L_2(c) = d_1^{-1}\circ \Phi^{-1}(c),
	\end{equation*}
	where $d_1^{-1}(\cdot)$ is defined in Eq.~\eqref{eq:d1inv}.
\end{prop}
\begin{pf}
It immediately follows from the monotonicity of $\Phi\circ d_1(\cdot)$ and the inequality,
$$ c = \Phi(d_1) - e^\kk \Phi(d_2) \le \Phi\circ d_1(\sigma). \quad \square$$
\end{pf}
The proof uses a lower bound of delta: $c\le \bsdelta(\sigma)$. The bounds of delta play a key role in interpreting the results of \citep{tehranchi2016bound} and discovering new results in this paper.

In the remarks below, we discuss the merits of $L_2(c)$ not discovered in \citep{tehranchi2016bound}. \citet[Prop.~4.6]{tehranchi2016bound} introduced another lower bound derived from the upper bound, $U_1(c)$:
\begin{equation} \label{eq:l_inv}
	\frac{2\kk}{U_1(1 - c)} = \frac{-\kk}{\Phi^{-1}\left(\frac{c}{1+e^\kk}\right)} \le \sigma.
\end{equation}
However, this lower bound is redundant because it is always inferior to $L_2(c)$. The next two remarks prove it.
\begin{rem}
The lower bound $L_2(c)$ satisfies an interesting symmetry,
$ L_2(c)\, L_2(1 - c) = 2\kk$,
because
$$ L_2(1 - c) = d_1^{-1}\circ \Phi^{-1}(1 - c) = d_1^{-1}(-\Phi^{-1}(c)) = \frac{2\kk}{\Phi^{-1}(c) + \sqrt{\Phi^{-1}(c)^2 + 2\kk}} = \frac{2\kk}{L_2(c)}.
$$
This symmetry provides an alternative proof of the inequality in \citep[Remark~4.7]{tehranchi2016bound}:
$$ 2\kk = L_2(c)\, L_2(1 - c) \le \sigma(c)\, \sigma(1 - c).$$
\end{rem}
\begin{rem}
	The bound $L_2(c)$ is tighter than any lower bound in the form of $2k/U(1-c)$, including that in Eq.~\eqref{eq:l_inv}. For any upper bound $U(c)$, the symmetry of $L_2(c)$ implies
	$$ \frac{2\kk}{U(1 - c)} \le \frac{2\kk}{L_2(1 - c)} = L_2(c) \le \sigma.
	$$
	Therefore, $L_2(c)$ can also be used for the limit of $c\downarrow 0$ with a fixed $\kk$~\citep[Prop.~4.6 and Thm.~4.4]{tehranchi2016bound} and for the limit of $c\downarrow 0$ and $\kk\uparrow\infty$~\citep[Prop.~4.9 and Thm.~4.8]{tehranchi2016bound}. In contrast, Proposition~\ref{p:up2lo} later offers a transformation that changes an upper bound into a lower bound that is always tighter than $L_2(c)$.
\end{rem}
\begin{rem}
	In addition to the $c\downarrow 0$ limits, $L_2(c)$ is also an accurate approximation for the $c\uparrow 1$ limit. Based on the proof, $L_2(c)$ is accurate when $\Phi(d_1) \gg e^\kk \Phi(d_2)$, which is when $\sigma\uparrow \infty$ or $c\uparrow 1$. Indeed, $L_2(c)$ is a better lower bound than $L_1(c)$, which has no dependency on $k$, as $c\uparrow 1$ for large $k$ values (see Figure~\ref{fig:bound_p}).
\end{rem}

We switch focus to the upper bounds. \citet[Thm.~3.1]{tehranchi2016bound} stated a generic method of constructing an upper bound. We present the theorem without proof, but reorganized via the option delta.
\begin{prop}[\mbox{\citet[Thm.~3.1]{tehranchi2016bound}}] \label{p:up_d1d2} The following is an upper bound of $\sigma$:
\begin{equation} \label{eq:FD}
	H(\mathcal{D}) = \Phi^{-1}(\mathcal{D}) - \Phi^{-1}\left(\frac{\mathcal{D} - c}{e^\kk}\right) \qtext{for} c < \mathcal{D} < 1.
\end{equation}
\end{prop}
\begin{rem} \label{r:HD}
In this presentation, $\mathcal{D}$ is understood as an estimation of delta, $\bsdelta(\sigma)$, because $\sigma=H(\mathcal{D})$ when $\mathcal{D} = \bsdelta(\sigma)$. Moreover, \citep[p.~914]{tehranchi2016bound} stated that $H\circ \bsdelta(\cdot)$ is a transformation that changes an upper bound $U(c)$ into a tighter bound:
$$ \sigma \le H\circ \bsdelta(U(c)) \le U(c).$$
\end{rem}
\begin{cor}[\mbox{\citet[Prop.~4.12]{tehranchi2016bound}}] \label{c:inf} Below, $U_2(\sigma)$ is an upper bound of $\sigma$:
	$$ U_2(c) = \Phi^{-1}\left(c + e^\kk \Phi(-\sqrt{2\kk})\right) + \sqrt{2\kk} \qtext{for} 0<c<1 - e^\kk \Phi(-\sqrt{2\kk})$$
\end{cor}
\begin{pf}
	Interpreting the proof in \citep{tehranchi2016bound} in terms of the delta bound, we obtain the following upper bound of delta:
	$$\bsdelta(\sigma) = \Phi(d_1) = c + e^\kk \Phi(d_2) \le c + e^\kk \Phi(-\sqrt{2\kk}).
	$$
	From this upper bound, we obtain $U_2(c) = H(c + e^\kk \Phi(-\sqrt{2\kk}))$. The equality holds at the inflection point, $\sigma=\sqrt{2k}$, because $d_2 = -\sqrt{2\kk}$ and $c + e^\kk \Phi(-\sqrt{2\kk})$ is the true delta when $\sigma=\sqrt{2k}$. $\square$
\end{pf}

\subsection{New bounds} \label{ss:new} \noindent
We obtain the tighter upper bounds of $\sigma$ by applying new upper bounds of delta to Proposition~\ref{p:up_d1d2}.
\begin{prop}[A tighter upper bound] \label{p:up_new} Below, $U_3(c)$ is an upper bound of $\sigma$:
	$$ U_3(c) = -\Phi^{-1}\left(\frac{1 - c}{2}\right) - \Phi^{-1}\left(\frac{1 - c}{2e^\kk}\right).
	$$ 
	Moreover, $U_3(c)$ is tighter than $U_1(c)$ (i.e., $\sigma \le U_3(c) \le U_1(c)$).
\end{prop}
\begin{pf}
	The delta, $\Phi\circ d_1(\sigma)$, is monotonically decreasing in $\kk$. Since we only consider $\kk\ge 0$, delta takes the maximum value $\Phi(\sigma/2)$ when $\kk=0$. Because $c = 2\,\Phi(\sigma/2) - 1$ when $\kk=0$, we obtain the upper bound of delta:
	$$ \bsdelta(\sigma) = \Phi\circ d_1(\sigma) \le \Phi(\sigma/2) = \frac{1 + c}{2}.
	$$ 
	We obtain $U_3(c) = H(\frac{1 + c}{2})$.
	To compare with $U_1(c)$, consider a function:
	$$ F(x) = -\Phi^{-1}\left(\frac{1-c}{2\,x}\right) \quad (x\ge 1).
	$$
	This function is concave on $x\ge 1$ for any $0<c<1$. Therefore,
	$$ U_3(c) = F(1) + F(e^{\kk}) \le 2F\left(\frac{1+e^{\kk}}{2}\right) = U_1(c). \quad \square
	$$
\end{pf}

\begin{cor} \label{c:u23}
	Both $U_2(c)$ and $U_3(c)$ are derived from the delta upper bounds; thus, we can merge $U_2(c)$ and $U_3(c)$, to generate an even tighter upper bound of $\sigma$:
	$$ U_{23}(c) = H\left(\min\left(\frac{1 + c}{2},\, c + e^\kk \Phi(-\sqrt{2\kk})\right)\right),
	$$
	where $H(\cdot)$ is defined in Eq.~\eqref{eq:FD}. Equivalently, $U_{23}(c) = \min(U_2(c), U_3(c))$.
\end{cor}
\begin{rem}
	As a minimum is taken along with $(1 + c)/2 < 1$, the restriction of $c$ in $U_2(c)$ (i.e., $0 < c < 1 - e^\kk \Phi(-\sqrt{2\kk})$) is no longer necessary in $U_{23}(c)$.
\end{rem}

Next, we explore the lower bounds of $\sigma$. These results depend on the following lemma on the ratio of the Black--Scholes price $\price(\sigma)$ to delta $\bsdelta(\sigma)$. The lemma is related to the Mills ratio $R(x)$, which is a heavily studied subject (see \citep{baricz2008mills} and \citep{sampford1953inequalities}):
$$ R(x) = \frac{\Phi(-x)}{\phi(x)}.$$
It is well known that $R(x)$ is a strictly decreasing and convex function of $x$~\citep{baricz2008mills}.

\begin{lemma} \label{lem:1} The price-to-delta ratio defined by
	$$ \pdelta(\sigma) = \frac{\price(\sigma)}{\bsdelta(\sigma)} = 1 - e^\kk \frac{\Phi\circ d_2(\sigma)}{\Phi\circ d_1(\sigma)} = 1 - \frac{R(-d_2(\sigma))}{R(-d_1(\sigma))},
	$$
	is monotonically increasing on $\sigma>0$ and monotonically decreasing on $\kk\ge 0$.
\end{lemma}
\begin{pf}
Note that the last equality is because of $\phi(d_1) = \phi(d_2)\, e^\kk$.
We prove the monotonicity in $\sigma$ by showing that the derivative of $R(-d_2)/R(-d_1)$ with respect to (w.r.t.) $\sigma$ is negative. The derivative is
\begin{align*}
	\frac{\partial}{\partial \sigma} \frac{R(-d_2)}{R(-d_1)} &=
\frac{d_1 R'(-d_2) R(-d_1) - d_2 R'(-d_1) R(-d_2)}{\sigma\, R(-d_1)^2} = - \frac{d_1 R(-d_1) - d_2 R(-d_2)}{\sigma\, R(-d_1)^2}
\end{align*}
where we used $d_1'(\sigma) = -d_2/\sigma$, $d_2'(\sigma) = -d_1/\sigma$, and $R'(x) = x R(x) - 1$. Because $R(-x)$ monotonically increases on $x$, so does $x R(-x)$. Moreover, $d_1(\sigma) > d_2(\sigma)$; thus, the derivative w.r.t. $\sigma$ is negative.

We prove the monotonicity in $\kk$ by showing that the derivative of $R(-d_2)/R(-d_1)$ w.r.t. $\kk$ is positive:
\begin{align*}
	\frac{\partial}{\partial \kk} \frac{R(-d_2)}{R(-d_1)} &=
	\frac{R'(-d_2) R(-d_1) - R'(-d_1) R(-d_2)}{\sigma\, R(-d_1)^2} \\
	&= \frac{R(-d_2)}{R(-d_1)} \left(1 + \frac{1}{\sigma R(-d_1)} - \frac{1}{\sigma R(-d_2)}\right) \\
	&= \frac{R(-d_2)}{R(-d_1)} \left(1 + \frac{d_2-d_1}{\sigma} m'(x) \right)
	= \frac{R(-d_2)}{R(-d_1)} \left(1 - m'(x) \right),
\end{align*}
where we apply the mean value theorem to $m(x) = 1/R(x)$ and $-d_1\le x\le -d_2$.
\citet[Eq.~(3)]{sampford1953inequalities} proved that $0 < m'(x) < 1$. Therefore, the derivative w.r.t. $\kk$ is positive. $\square$
\end{pf}
Figure~\ref{fig:price} (upper right) displays the shape of $\pdelta(\sigma; k)$. 
\begin{prop} \label{p:up2lo} Let
	\begin{equation*}
		J(y) = d_1^{-1}\circ \Phi^{-1}\left(\frac{c}{\pdelta(y)}\right).
	\end{equation*}
	When $U(c)$ is an upper bound of $\sigma$, $J\circ U(c)$ is a lower bound, and is always tighter than $L_2(c)$ in Proposition~\ref{p:l2}.
\end{prop}
\begin{pf}
	That $J\circ U(c)$ is a lower bound is an immediate result from Lemma~\ref{lem:1}:
	$$ \frac{c}{\bsdelta(\sigma)} = \pdelta(\sigma) \le \pdelta(U(c)) \qtext{or}
	\frac{c}{\pdelta(U(c))} \le \bsdelta(\sigma). $$
	The bound $J\circ U(c)$ is tighter than $L_2(c)$ because $\pdelta(\cdot) \le 1$ and $d_1^{-1}\circ \Phi^{-1}(\cdot)$ is monotonically increasing:
	$$ L_2(c) = d_1^{-1}\circ \Phi^{-1}(c) \le d_1^{-1}\circ \Phi^{-1}\left(\frac{c}{\pdelta(U(c))}\right) = J\circ U(y) \le \sigma. \quad \square$$
\end{pf}
\begin{rem}
	Conversely, when $L(c)$ is a lower bound, $J\circ L(c)$ is an upper bound. However, we did not explore this direction because we focused on determining tighter lower bounds for root-solving in Section~\ref{s:root}. It would be interesting to observe how the two upper bounds, $J\circ L(c)$ and $H\circ \bsdelta(L(c))$, obtained from a lower bound $L(c)$, compete.
\end{rem}
Based on Proposition~\ref{p:up2lo}, we provide two new lower bounds.
\begin{cor}[First tighter lower bound] \label{c:l23} From the upper bound $U_{23}(c)$ in Corollary~\ref{c:u23}, we obtain a new lower bound:
\begin{equation*}
L_\mathrm{U23}(c) = J\circ U_{23}(c).
\end{equation*}
\end{cor}
The next lower bound uses a similar idea, although it is not exactly in the form of Proposition~\ref{p:up2lo}.
\begin{cor}[Second tighter lower bound] \label{c:l3}
	Below, $L_3(c)$ is also a lower bound of $\sigma$ and is tighter than $L_2(c)$:
	\begin{equation*}
		L_3(c) = d_1^{-1} \circ \Phi^{-1}\left(\frac{c\,(c+e^\kk)}{2c + e^\kk - 1} \right).
	\end{equation*}
\end{cor}
\begin{pf} Using $U_1(c)$ in Proposition~\ref{p:teh} and the monotonicity of $\pdelta$ in $k$ in Lemma~\ref{lem:1}, we can bound $\pdelta(\sigma)$:
$$ \pdelta(\sigma) \le \pdelta(U_1(c)) \le \pdelta(U_1(c); \kk=0) = \frac{2\,\Phi(U_1(c)/2)-1}{\Phi(U_1(c)/2)} = \frac{2c + e^\kk - 1}{c+e^\kk} \le 1.
$$
Therefore, we obtain the lower bound:
$$ L_3(c) = d_1^{-1} \circ \Phi^{-1}\left(\frac{c}{\pdelta(U_1(c); \kk=0)}\right) = d_1^{-1} \circ \Phi^{-1}\left(\frac{c\,(c+e^\kk)}{2c + e^\kk - 1}\right).
$$
Although not in the form of $J\circ U(c)$, $L_3(c)$ is also tighter than $L_2(c)$ because $\pdelta(U_1(c); \kk=0) \le 1$. $\square$
\end{pf}
Several remarks are in order.
\begin{rem}
	As $L_3(c)$ is tighter than $L_2(c)$, it provides better asymptotics than $L_2(c)$ as $c\downarrow 0$ and $c\uparrow 1$ (see Figure~\ref{fig:bound_p}). Moreover, unlike $L_2(c)$, $L_3(c)$ becomes exact when $k=0$ (ATM). Therefore, all bounds discussed so far, except $L_2(c)$ and $U_2(c)$, are exact if $k=0$:
	$$L_1(c) = L_3(c) = L_{U23}(c) = \sigma = U_{23}(c) = U_3(c) = U_1(c) \qtext{if} \kk=0.$$
\end{rem}
\begin{rem}
The bound $L_\mathrm{U23}(c)$ is tighter than $L_3(c)$ because
$\pdelta(U_{23}(c)) \le \pdelta(U_3(c)) \le \pdelta(U_1(c)) \le \pdelta(U_1(c);\kk=0).$
Therefore, the order of the lower bounds is $L_2(c) \le L_3(c) \le L_\mathrm{U23}(c) \le \sigma$.
\end{rem}
\begin{rem}
	From the proofs of the Cors. \ref{c:u23} and \ref{c:l3}, we obtain the bounds of the delta as functions of the option price $c$:
	$$ c \le \frac{c\,(c+e^\kk)}{2c + e^\kk - 1} \le \bsdelta(\sigma) \le \min\left(\frac{1 + c}{2},\, c + e^\kk \Phi(-\sqrt{2\kk})\right).
	$$
\end{rem}

\begin{figure}[ht!]
	\caption{\label{fig:bound_p} Upper and lower bounds and exact IV ($\sigma$) as functions of price $c$ for $e^\kk= 2$ (left) and $e^\kk= 8$ (right). The blue circle marks the inflection point $(\price(\sqrt{2k}),\,\sqrt{2k})$. While the upper plots cover the whole range of $c$ values, the lower plots focus on the $c\uparrow 1$ (high $\sigma$) limit. In the lower plots, the price is transformed by $L_1(c)$ to present the asymptotics better.}
	\begin{center}
		\includegraphics[width=0.48\linewidth]{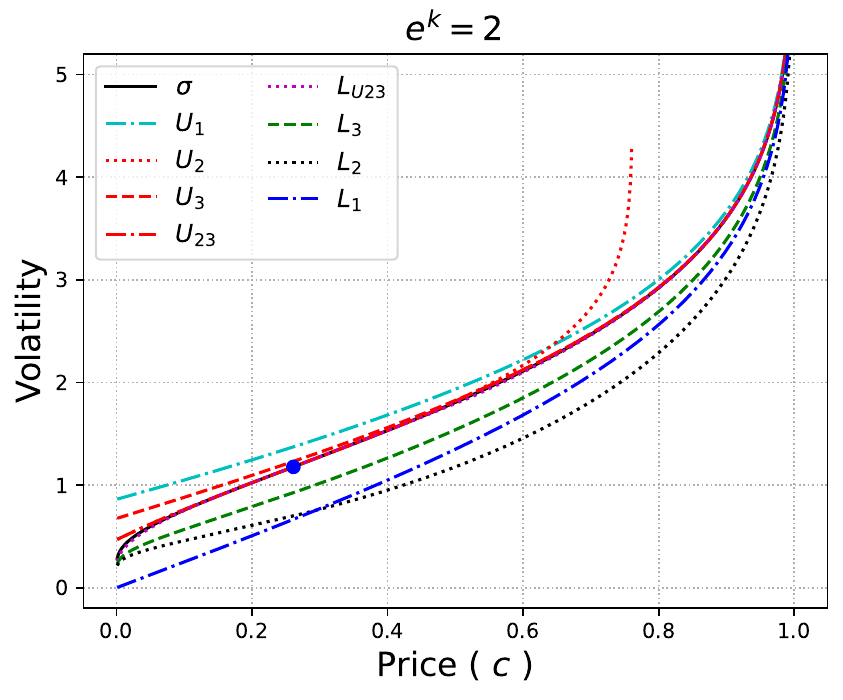}
		\includegraphics[width=0.48\linewidth]{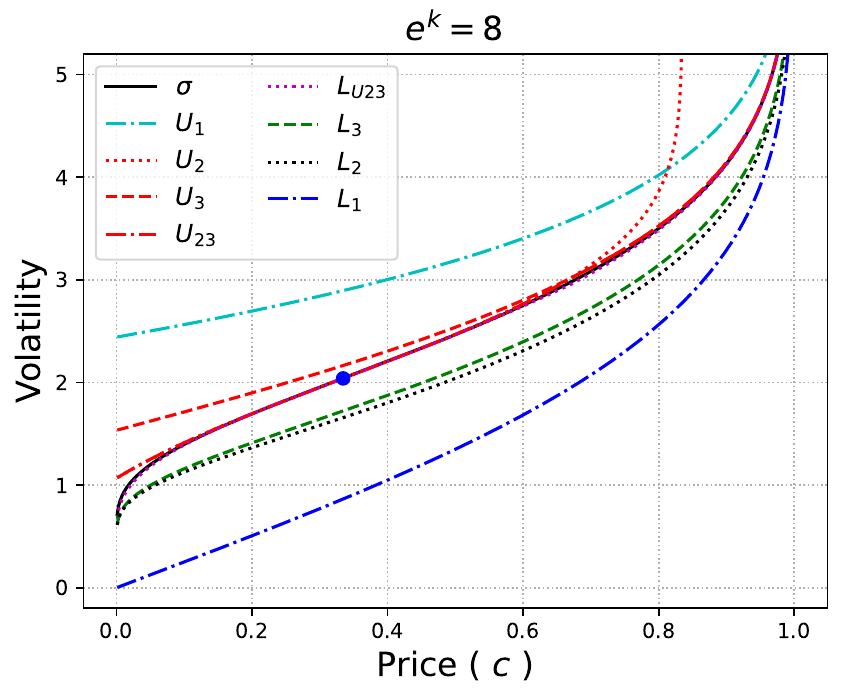}\\
		\includegraphics[width=0.48\linewidth]{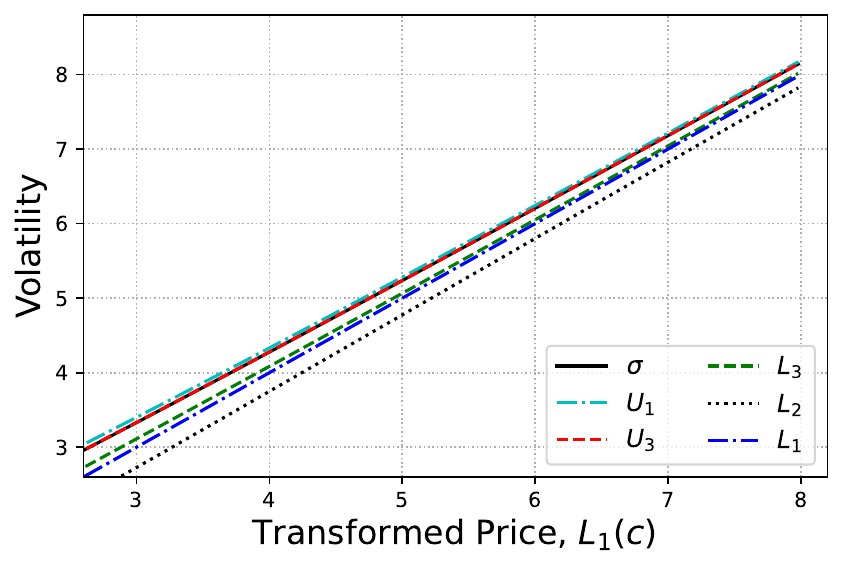}
		\includegraphics[width=0.48\linewidth]{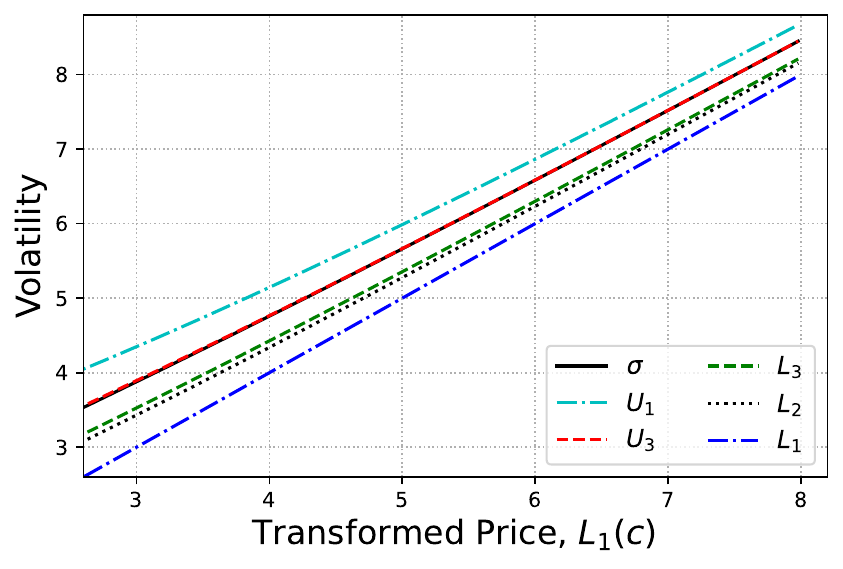}
	\end{center}
\end{figure}
\begin{figure}[ht!]
	\caption{\label{fig:bound_sig} Upper and lower bounds implied from $c=\price(\sigma; k)$ as functions of $k$ with fixed $\sigma=0.2$ (left) and $\sigma=1.5$ (right). The blue circle marks the inflection point, $k=\sigma^2/2$.}
	\begin{center}
		\includegraphics[width=0.48\linewidth]{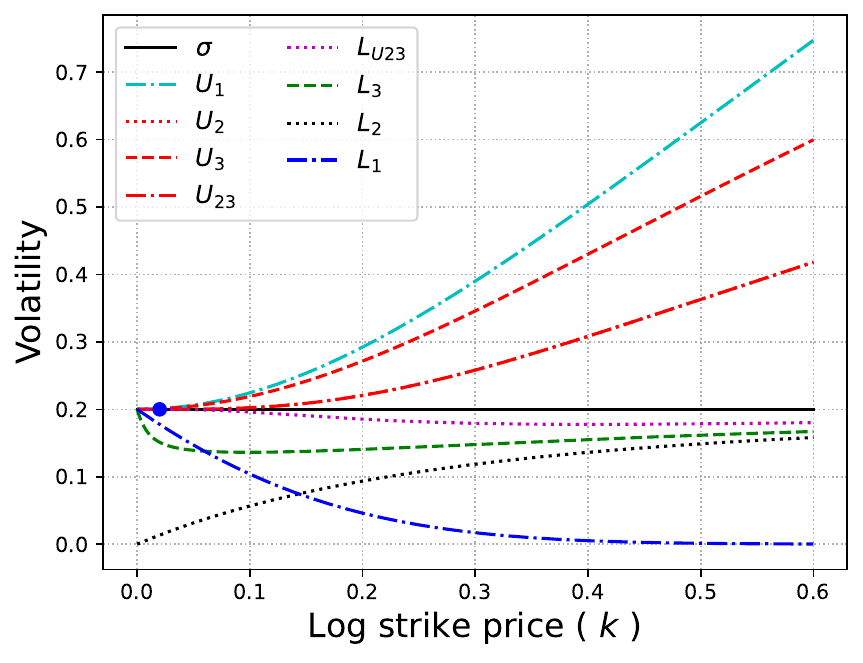}
		\includegraphics[width=0.48\linewidth]{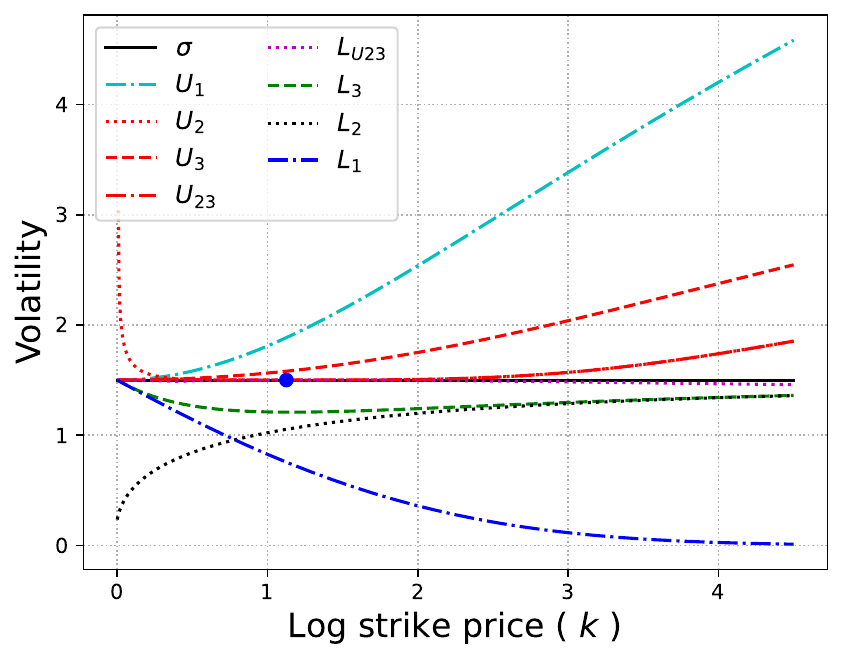}
	\end{center}
\end{figure}

Figures~\ref{fig:bound_p} and \ref{fig:bound_sig} depict the lower and upper bounds discussed in this section and confirm the tightness of the new bounds. The new upper bound $U_3(c)$ is a very accurate approximation of $\sigma$ as $c\uparrow 1$. The bottom panel of Figure~\ref{fig:bound_p} uses $L_1(c)$, instead of $c$, in the $x$-axis so that the asymptotic behavior is roughly linear. The bound $U_{23}(c)$ further extends the validity region below the inflection point. The lower bound $L_{U23}(c)$ transformed from $U_{23}(c)$ is very close to $\sigma$ across all regions of $0<c<1$. The new lower bound $L_3(c)$ is tighter than $L_2(c)$. Although $L_3(c)$ is a better approximation than $L_1(c)$ in general, $L_3(c)$ is not always tighter than $L_1(c)$ because exceptions exist near $k\approx 0$ as illustrated in Figure~\ref{fig:bound_sig}.

\section{Applications to numerical root-finding of implied volatility} \label{s:root} \noindent
In this section, we propose a new NR iteration on the log price that uses one of the new lower bounds as an initial guess. Inverting the option price to IV is arguably one of the most heavily used computations in finance. Given the absence of an analytical inversion formula, there has been extensive research on this topic. See \citep{orlando2017review} for a review. Among them, \citep{jackel2015let} and \citep{glau2019chebyshev} are known as the most efficient algorithms. 
Although not the most efficient, the naive NR method~\citep{manaster1982iv} (see Section~\ref{ss:nr_naive}) has been widely adopted because of simple implementations and guaranteed convergence. 
We present our new algorithm in Section~\ref{ss:nr_log} as a better alternative to the naive NR method. While the iteration formula (Proposition~\ref{p:nr_log}) is simple, the convergence is fast and uniform over all price ranges, as we demonstrate. 

\subsection{Naive Newton--Raphson method} \label{ss:nr_naive} \noindent
The IV, $\sigma$, is the root of
$$ f(y) = \price(y) - c = 0. $$
The NR iteration applied to the equation is given by 
$$ \sigma_{n+1} = \sigma_n - \frac{f(\sigma_n)}{f'(\sigma_n)} = \sigma_n - \frac{\price(\sigma_n) - c}{\bsvega(\sigma_n)}.$$
As $\price(y)$ is convex (concave) below (above) the inflection point, the NR iteration always converges to the root $\sigma$ if the iterations begin from the inflection point $\sigma_0 = \sqrt{2\kk}$~\citep{manaster1982iv}. We call this approach the \textit{naive} NR method.

The naive NR method usually works well, but the convergence slows significantly for the far out-of-the-money options (i.e., $c\downarrow 0$ with $k>0$) because $\price(y)$ flattens to zero below some level of $y$. Because $e^{-k^2/2\sigma^2}$ has an essential singularity at $\sigma=0$, the derivatives of $\price(\sigma)$ in any order are zero. For example, we consider an option with $e^\kk = 1.5$ and $\sigma=0.04$. This example is borrowed from the Wilmott forum thread titled ``\href{https://forum.wilmott.com/viewtopic.php?f=34&t=97812&start=37}{iv for all and all for iv}''.
If we invert the corresponding option price $c\approx 9.01\times 10^{-27}$ with the naive NR method, 50 iterations barely reach $\sigma_{50} = 0.04173$, which is still far from $\sigma=0.04$. Although it is difficult to encounter such a low option price in the option markets, accurate IV in extreme cases is required for IV asymptotics and calibration of advanced option pricing models.

\subsection{New Newton--Raphson methods on the log price} \label{ss:nr_log} \noindent
A natural solution to the problem caused by a low price is considering the log objective function:
\begin{equation} \label{eq:root_log}
	g(y) = \log \price(y) - \log c = 0.
\end{equation}
The objective function using log is not completely new. \citet{jackel2015let} used the reciprocal log price,  $g(y) = 1/\log\price(y) - 1/\log c$, and the log complimentary price, $g(y) = \log(1-\price(y)) - \log (1-c)$ for very low and very high price regions, respectively. To the best of our knowledge, the objective function in Eq.~\eqref{eq:root_log} with simple log price is new in the literature, and we use it for all price ranges uniformly.

As depicted in Figure~\ref{fig:price} (lower left), the log price changes rapidly as $\sigma\downarrow 0$; hence, it is suitable for the NR method. The question is whether the log NR method also guarantees convergence, as in the naive NR method. The following lemma, which parallels Lemma~\ref{lem:1}, provides an important clue.
\begin{lemma} \label{lem:2} The price-to-vega ratio $\pvega(\sigma)$ defined by
$$ \pvega(\sigma) = \frac{\price(\sigma)}{\bsvega(\sigma)} = \frac{\Phi\circ d_1(\sigma) - e^\kk \Phi\circ d_2(\sigma)}{\phi\circ d_1(\sigma)} = R(-d_1(\sigma)) - R(-d_2(\sigma)),
$$
is a monotonically increasing function of $\sigma > 0$.
\end{lemma}
\begin{pf}
	The derivative w.r.t. $\sigma$ is
	\begin{align*}
		\pvega'(\sigma) &= \frac{1}{\sigma} [R'(-d_1) d_2 - R'(-d_2) d_1] = \frac{d_2}{\sigma} [R'(-d_1) - R'(-d_2)]
		- R'(-d_2).
	\end{align*}
	This result is positive because $R(x)$ is decreasing (i.e., $R'(-d_2)<0$) and convex $R'(-d_1) < R'(-d_2)$, and $d_2<0$. $\square$
\end{pf}
Figure~\ref{fig:price} (lower right) presents the shape of $\pvega(\sigma; k)$. 
\begin{prop} \label{p:nr_log}
The NR iteration for the log objective function, Eq.~\eqref{eq:root_log}, is expressed by
$$ \sigma_{n+1} = G(\sigma_n) \qtext{where}
G(y) = y + \left[\frac{d_1(y)^2}{2} - \log \pvega(y) + \log(c \sqrt{2\pi})\right] \pvega(y).
$$
The iteration converges to the root $\sigma$ if a lower bound is used as the initial guess, $\sigma_0 \le \sigma$.
\end{prop}
\begin{pf}
The iteration step follows from the NR formula, $ \sigma_{n+1} = \sigma_n - g(\sigma_n)/g'(\sigma_n)$, where
\begin{equation*}
	g(y) = \log\phi(d_1(y)) + \log \pvega(y) - \log c
	\qtext{and}
	g'(y) = \frac{\bsvega(y)}{\price(y)} = \frac1{\pvega(y)}.
\end{equation*}
From Lemma~\ref{lem:2}, $g(y)$ is strictly increasing and concave (i.e., $\price(\sigma)$ is log-concave). Therefore, the NR iteration from a lower bound ($\sigma_0\le \sigma$) always converges to $\sigma$ from below. $\square$
\end{pf}
\begin{rem}
	The volatility transformation $G(y)$ makes a lower bound tighter.
	It is a counterpart of $H\circ \bsdelta(y)$ (see Remark~\ref{r:HD}), which makes an upper bound tighter.
\end{rem}

\subsection{Numerical Experiment} \noindent
We discuss the numerical implementation and results of the new root-finding method in Proposition~\ref{p:nr_log}. Among many lower bounds of $\sigma$ in Section~\ref{s:bounds}, we select $L_3(c)$ in Corollary~\ref{c:l3} for the initial guess $\sigma_0$. This choice is a compromise between tightness and computational cost. Although $L_\mathrm{U23}(c)$ is the tightest lower bound, the computation cost is high. Whereas $L_3(c)$ requires only one evaluation of $\Phi^{-1}(\cdot)$, $L_\mathrm{U23}(c)$ requires the evaluation of $\Phi^{-1}(\cdot)$ three times (in $U_{23}$ and $G$) and $\Phi(\cdot)$ two times (in $\pdelta$). The extra computational cost for $L_\mathrm{U23}(c)$ roughly corresponds to two more iterations, and $L_3(c)$ improved by two iterations (i.e., $\sigma_2 = G\circ G\circ L_3(c)$) is better than $L_{U23}(c)$ overall.

Evaluating $G(y)$ requires only $\pvega(y)$ and other elementary functions.
In the implementation, we evaluated $\pvega(\cdot)$ via the scaled complementary error function, $\erfcx{x}$,
$$ \erfcx{x} = \frac{2e^{x^2}}{\sqrt\pi}\int_x^\infty e^{-t^2}dt \qtext{and}
R(x) = \sqrt{\frac{\pi}2} \erfcx{\frac{x}{\sqrt2}},$$
because a fast evaluation algorithm~\citep{johnson2015faddeeva} is available in Python \texttt{scipy} package. Therefore, the computational cost of one iteration in the log NR method is similar to that in the naive NR method. 

\begin{figure}[ht!]
\caption{\label{fig:conv} Convergence of the naive NR (top) and log NR (bottom) methods for $e^k = 2$ (left) and $e^k = 8$ (right).}
\begin{center}
	\includegraphics[width=0.48\linewidth]{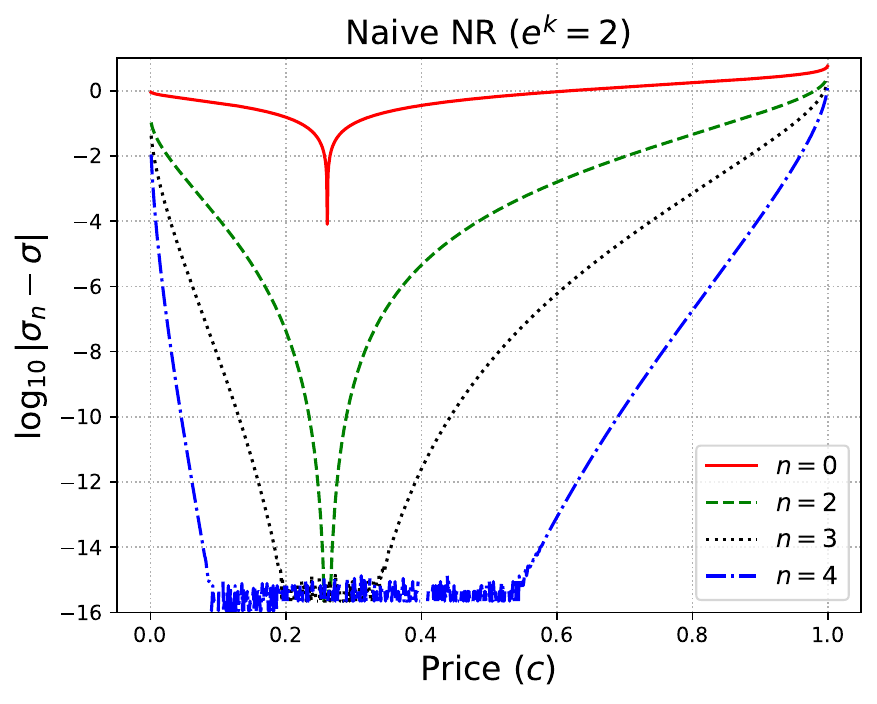}
	\includegraphics[width=0.48\linewidth]{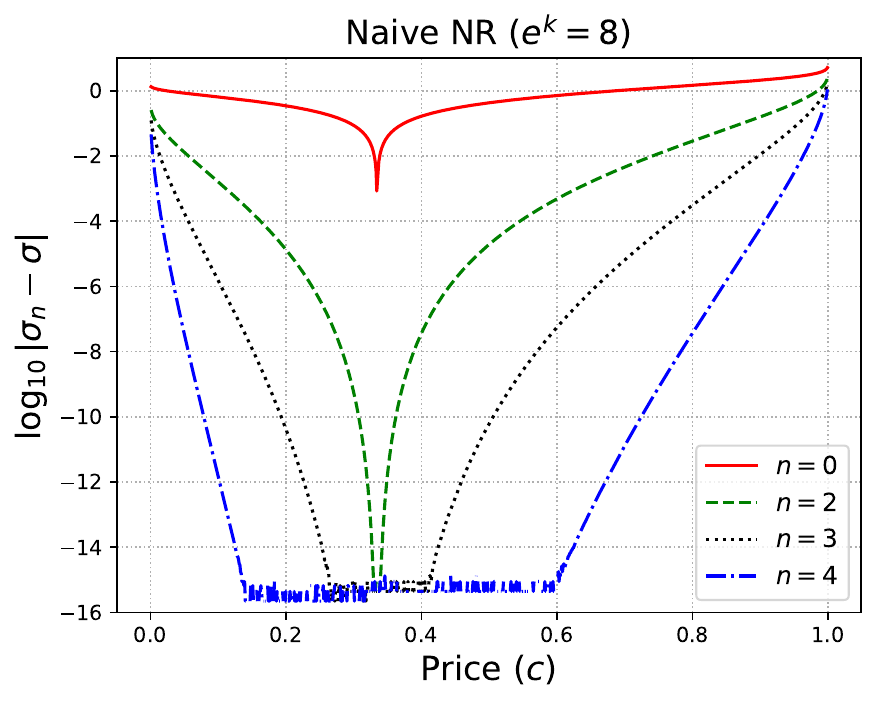}\\
	\includegraphics[width=0.48\linewidth]{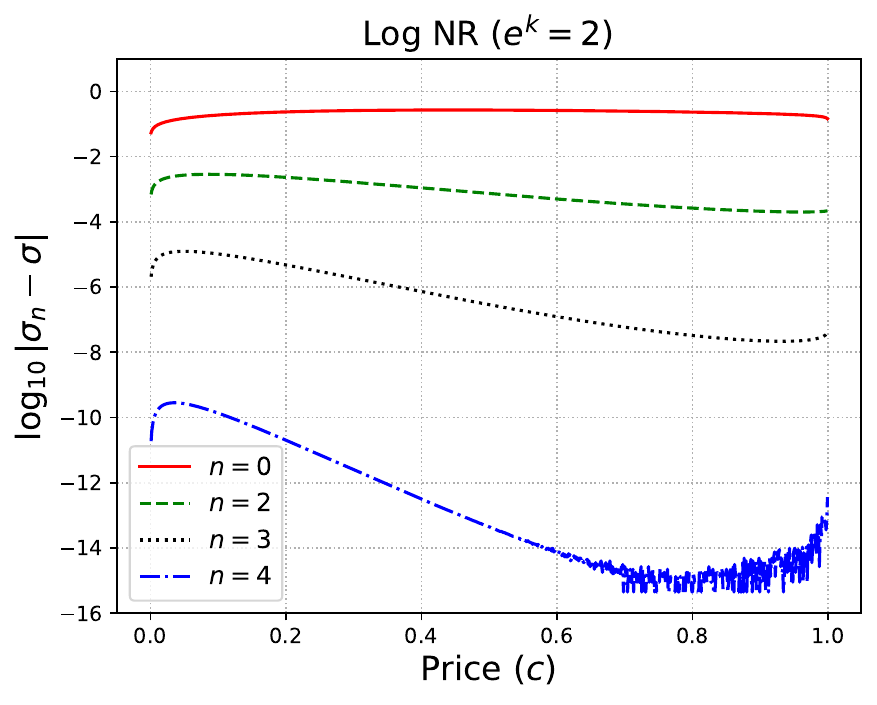}
	\includegraphics[width=0.48\linewidth]{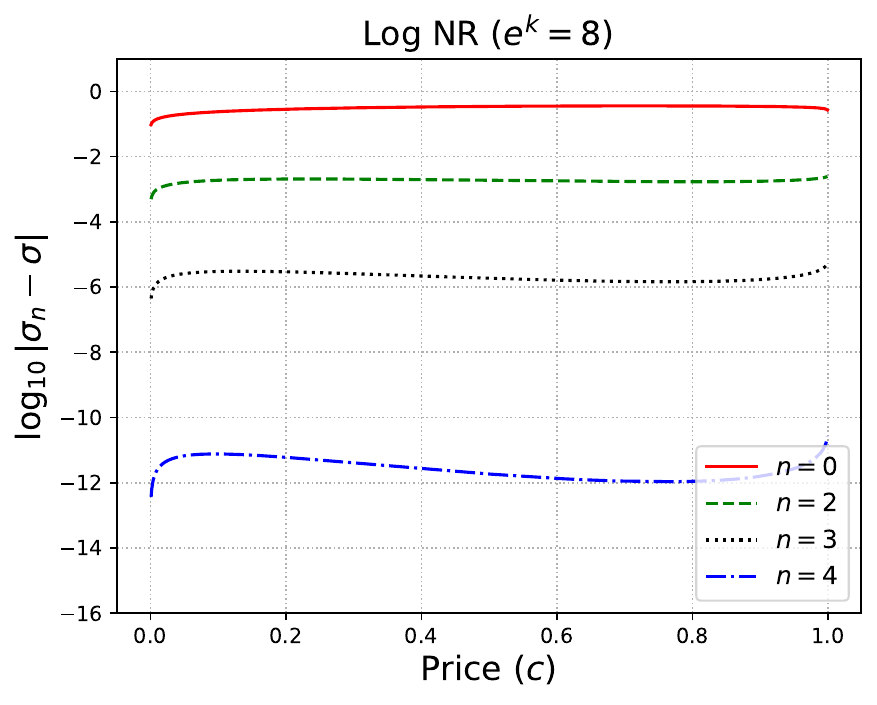}
\end{center}
\end{figure}

We test the accuracy of the new NR method for an extensive range of the grid of $(k, c)$:
\begin{align*}
	k &\in \{0, 10^{-10},\; 10^{-9}, \cdots, 0.01,\; 0.02,\; 0.03, \cdots, 10.00 \}, \\
	c &\in \{10^{-40},\; 10^{-39}, \cdots, 0.0001,\; 0.0002,\; 0.0003, \cdots, 0.9999 \}.
\end{align*}
After $n$ iterations, we measure the log error, $|g(\sigma_n)|$, which is much stricter than the absolute error, $|f(\sigma_n)|$:
$$ |g(\sigma_n)| = |\log \price(\sigma_n) - \log c\,| \approx \frac{|\price(\sigma_n) - c\,|}{c} \ge |f(\sigma_n)|.$$
The maximum log error of the new NR method over the $(k,c)$ grid is in the order of $10^{-8}$ after four iterations and in the order of $10^{-13}$ after five iterations. The absolute price error corresponds to the orders of $10^{-11}$ and $10^{-16}$, respectively. The maximum difference between $\sigma_4$ and $\sigma_5$ is $7\times 10^{-10}$. For the example discussed above ($e^\kk = 1.5$, $\sigma=0.04$, and $c\approx 9.01\times 10^{-27}$), only three iterations result in the volatility error, $\sigma_3 - \sigma = -2\times 10^{-11}$, and the log price error, $g(\sigma_3) = - 10^{-8}$.

Figure~\ref{fig:conv} compares the convergence speed of the naive and log NR methods for the whole price range. In the naive NR method, the convergence is slow when the price moves to both ends of the extreme. For a significant portion of low and high prices, $\sigma_n$ does not reach high accuracy after four iterations. In the log NR method, however, the convergence is uniformly fast. At each stage of iterations, the volatility error level is similar across the price range, and the error level falls below $10^{-9}$ after four iterations for all prices in the two examples. The fast and uniform convergence of the log NR method is thanks to the log objective function and tje tight lower bound, $L_3(c)$, used for an initial guess of the log NR method.

\newpage
\section*{Funding} \noindent
Jeonggyu Huh received financial support from the National Research Foundation of Korea (Grant No. NRF-2022R1F1A1063371).

\section*{Data Availability Statement} \noindent
The Python implementations that produce the findings of this study are openly available at \url{https://github.com/PyFE/PyfengForPapers}.

\begin{singlespace}
	\bibliography{BlackScholes-Abbr}
\end{singlespace}
\end{document}